\documentclass{article}
\usepackage[margin=1in]{geometry}
\usepackage{amsmath,amssymb,amsthm}
\usepackage{url}

\newtheorem{theorem}{Theorem}[section]
\newtheorem{lemma}[theorem]{Lemma}
\newtheorem{corollary}[theorem]{Corollary}
\newtheorem{definition}[theorem]{Definition}
\theoremstyle{remark}
\newtheorem{remark}[theorem]{Remark}
\newcommand{\F}{\mathbb F}
\newcommand{\Rad}{\operatorname{Rad}}
\newcommand{\Fix}{\operatorname{Fix}}
\newcommand{\Tr}{\operatorname{Tr}}
\newcommand{\ket}[1]{|#1\rangle}

\title{Exact quantum splitting and the structure of finite algebras}
\author{Muhammad Imran\\
School of Computer Science, University of Birmingham\\
CISITS Lab, Department of Mathematics, Universitas Indonesia\\
\texttt{m.imran.2@bham.ac.uk; m.imran@sci.ui.ac.id}}
\date{\today}

\begin{document}
\maketitle

\begin{abstract}
For squarefree $f\in\F_q[x]$, Berlekamp's fixed-point algebra is
$B\cong\F_q^r$, where $r$ is the number of irreducible factors. We give an
exact quantum splitting procedure in a circuit model permitting efficiently
computable, input-dependent single-qubit rotations. For a block with $s\ge2$
factors, a quadratic-character test (odd $q$) or absolute-trace test
(characteristic $2$) is nonconstant with an exactly known probability
$p_{q,s}\ge1/2$. Its value depends only on $q,s$, not on the hidden factors.
An auxiliary-qubit reduction to success probability $1/4$ makes one Grover
iteration exact. Thus factorization takes exactly $r-1$ splitting rounds
and $O(n^3\log q)$ quantum $\F_q$-operations, with deterministic polynomial
classical preprocessing and bookkeeping. The method also extracts primitive
idempotents of arbitrary finite-dimensional commutative unital algebras
given by structure constants in $O(n^4\log q)$ quantum
$\F_q$-operations, staying over $\F_q$ throughout. Computing the radical and passing to the
center extends it to the decomposition of $A/\Rad(A)$ into simple
two-sided ideals, for an arbitrary finite-dimensional associative
$\F_q$-algebra $A$. 
\end{abstract}

\section{Introduction}\label{sec:intro}
Berlekamp's algorithm computes the fixed-point algebra of the $q$-power
Frobenius on $\F_q[x]/(f)$ by deterministic linear algebra, then splits its
coordinates to recover the factors \cite{Berlekamp67,Berlekamp70}. For
fixed small characteristic, the splitting stage is already deterministic;
over general large odd fields, efficient classical methods use randomness.
Unconditional deterministic polynomial-time factorization in $n,\log q$
has long been open, even under the Extended Riemann Hypothesis in its full
generality \cite{Gao01}; a recent unrefereed claim over prime fields is
discussed below. Our aim is an exact quantum procedure for this
splitting stage, without a primitive root or non-residue subroutine.

\paragraph{Relation to a recent deterministic claim.}
A manuscript released by OpenAI on October~4, 2026 \cite{OpenAI26} claims
an unconditional deterministic algorithm that completely factors a
polynomial over a prime field $\F_p$ in
$O\bigl(((n+1)\lceil\log_2p\rceil)^{10^{12}}\bigr)$ bit operations. Its
analytic input, a uniform zero-free region for Hecke $L$-functions, is
taken from a companion manuscript in the same release. The first version
of the present paper was posted to arXiv on August~31, 2026
(arXiv:2608.30340v1), before that manuscript appeared, and the two works
are independent. According to the accompanying release, the manuscript
was produced by an AI system, and to our knowledge neither it nor its
companion has been peer reviewed; we do not rely on either. The questions
also differ. That manuscript concerns classical deterministic
computation over prime fields and replaces the random splitting choice by
an arithmetic construction based on auxiliary primes and cyclic covers.
The present paper concerns exact quantum computation over every finite
field $\F_q$, including extension fields, and over finite-dimensional
algebras given by structure constants. It obtains the splitting from a
counting argument with low-degree cost bounds. Should the deterministic
claim be confirmed, the comparison with classical computation in
Section~\ref{sec:comparison} would change for prime fields, but the
statements proved here would not.

\vspace{2cm}
The observation is simple: although the individual CRT coordinates are
hidden, a uniform element of the fixed-point algebra has independent uniform
coordinates. A test map with known fiber sizes therefore has a known
probability of separating them. This count is the standard one in the
analysis of random splitting \cite[Ch.~14]{vzGG13}; classically it bounds
the failure probability, whereas here its exact value is used. Exact
amplitude amplification with a known success probability
\cite{BH97,BHMT02,MZ03} then converts the test into a certain split.
Draper \cite{Draper21} observed that this device samples exactly from any
set of known cardinality that has an efficient membership test; applied to
the set of useful test elements, it gives one exact splitting round almost
immediately. The fixed-point approach to abstract algebras, the
computation of primitive idempotents of commutative algebras
\cite{Davis94}, and the reductions from structure problems of algebras to
polynomial factoring \cite{FR85,Ronyai90} are likewise classical.

Our contributions are therefore the following. (i) The relevant
cardinality depends only on $q$ and the current block dimension, not on the
hidden factors, so the same exact round applies at every node of an
adaptive recursion (Section~\ref{sec:main}). (ii) We give the circuit, the
model in which it is exact, and separate quantum and classical cost bounds
(Sections~\ref{sec:model} and~\ref{sec:main}). (iii) The same splitter
decomposes an arbitrary commutative unital algebra given by structure
constants by test elements and idempotent interpolation
(Section~\ref{sec:generalization}). It stays over $\F_q$ throughout; in
particular, in the semisimple case it avoids the factoring calls over auxiliary extension
fields that arise when an exact factoring oracle is substituted into the
classical cutting procedure. Passing to the center of the semisimple
quotient gives the decomposition of an associative algebra modulo its
radical into simple ideals (Section~\ref{sec:noncommutative}).

The claim concerns the uniform circuit model with exact gates specified in
Section~\ref{sec:model}. A finite gate set with approximate synthesis would
not give the same zero-error statement. In characteristic $2$, the
construction is also an exact quantum counterpart of an existing
deterministic classical method \cite{Berlekamp67}.

\paragraph{What is exact.}
Splits are classically verifiable: a candidate divisor is checked by one
polynomial division, and the complete output by the deterministic test of
Appendix~\ref{app:verify}. Repeating any bounded-error splitting step until
it succeeds therefore already gives a zero-error algorithm with polynomial
\emph{expected} running time, classically (ZPP) or quantumly (ZQP). The
present construction adds a worst-case bound. At each splitting node, the
quantum circuit is prepared without measurements, and its final measurement
has support only on useful test elements. The subsequent gcd, basis update,
and choice of the next node are deterministic given the measured element.
An individual output may vary, but every possible output is correct, and
the number of rounds is exactly $r-1$. The factorization is thus computed
with zero error in worst-case polynomial time, which places it in exact
quantum polynomial time in the sense of Section~\ref{sec:model}, not
merely in ZQP.


\section{Fixed-point algebras and exact amplification}\label{sec:prelim}
Let $f=f_1\cdots f_r$ be monic and squarefree, $\deg f=n\ge1$, with $f_i$
distinct and irreducible. Put $R=\F_q[x]/(f)$ and
\[
 B=\Fix(a\mapsto a^q)=\ker(\sigma-I)\subseteq R.
\]
The CRT and the fact that $\Fix(x\mapsto x^q,\F_{q^{d_i}})=\F_q$ give
\[
 R\cong\prod_{i=1}^r\F_{q^{d_i}},\qquad
 B\cong\F_q^r,\qquad \dim_{\F_q}B=r.                 \tag{1}\label{eq:crt}
\]
An $\F_q$-basis $v_1=1,v_2,\ldots,v_r$ is computed from the Frobenius
matrix without knowing the $f_i$. For a current squarefree divisor
$g_P=\prod_{i\in P}f_i$, write $n_P=\deg g_P$, $s=|P|$, and
\[
 R_P=\F_q[x]/(g_P),\qquad B_P=\Fix(a\mapsto a^q,R_P)\cong\F_q^s.
\]
The notation $P$ describes a hidden set; the algorithm stores $g_P$ and a
basis of $B_P$, not the unknown factors indexed by $P$. The dimension of
that basis supplies the known value $s$. An element of $B_P$ is constant
precisely when all its CRT coordinates agree.

To see the connection with factorization explicitly, let
$b\in B_P$ have CRT coordinates $(b_i)_{i\in P}$. For every $t\in\F_q$,
\[
 \gcd(g_P,b-t)=\prod_{\substack{i\in P\\b_i=t}}f_i.
\]
Thus any $b$ taking two values yields a proper factor once a value in
its image is known. A uniformly chosen $b$ is almost always
nonconstant, but its image could contain any of the $q$ field
elements. The character and trace maps below compress that image to
at most three known candidates while retaining a large, exactly
countable fraction of nonconstant outputs. No CRT isomorphism is
computed by the algorithm; it is used only in the analysis.

A nonzero constant input has an empty irreducible-factor list and needs
no splitting round; the zero polynomial is excluded. A nonconstant input
is first made monic, retaining its leading coefficient, and can then
undergo the standard deterministic
squarefree decomposition, including $p$-th-root recursion when $f'=0$;
$f/\gcd(f,f')$ alone is insufficient in positive characteristic.
The theorem's complexity statement concerns one squarefree input.
Applying it separately to the squarefree components also recovers
multiplicities; those preprocessing costs are not included in
\eqref{eq:quantumcost}.

\subsection{Circuit model}\label{sec:model}
The input includes a presentation $\F_q=\F_p[T]/(h)$, where $q=p^k$,
$p$ is prime, and $h$ is a supplied monic irreducible polynomial of degree
$k$. Field construction is not included in our bounds. Field elements
have binary encodings in this presentation, and field arithmetic is
performed reversibly. We count arithmetic in $\F_q$ separately from the classical
linear algebra and from elementary gate or bit costs. We allow exact
one-qubit rotations whose coefficients are computable to prescribed
precision in time polynomial in the input length and the precision.
The rotations depend on the input $q$ and the current, classically known,
block dimension $s$, as in the input-dependent uniform exact-circuit model
of \cite{imran2022exact,Imran-Ivanyos2022,NO05}. Exact gates, not numerical
approximations to them, are part of this model. In the notation of
Adleman, DeMarrais and Huang, who index exact quantum polynomial time by
the set $K$ of permitted gate coefficients \cite{ADH97}, this is
$\mathrm{EQP}_{\mathbb C}$ restricted by the uniformity condition of
Nishimura and Ozawa \cite{NO05}, with gate parameters computed from the
input rather than from its length; we write $\mathrm{EQP}^u_{\mathbb C}$
for this class. Every exactness claim below is made in it.

To prepare a uniform field register when $q$ is not a power of two,
encode its elements by $0,\ldots,q-1$. Binary interval preparation uses
controlled rotations determined by the numbers of valid completions of
each prefix, followed by a reversible change to the field representation.
It produces $q^{-1/2}\sum_{a\in\F_q}\ket a$ exactly with polynomially many
gates. Tensoring $s$ copies and forming $\sum_k\lambda_kv_k$ defines the
measurement-free preparation $A_P$. Appendix~\ref{app:preparation} gives
the construction and its amplitude calculation. The circuit is a fixed
unitary with an inverse, as the reflection in \eqref{eq:grover} requires;
repeat-until-accept sampling does not provide one.

Arithmetic-count statements refer to addition and multiplication in
$\F_q$ performed inside a reversible circuit. The elementary-gate cost of
an $\F_q$ operation depends on the field representation and the
multiplication circuit \cite{VBE96}, and is not included in these counts.

For the attenuation gate, write $p_{q,s}=N/D$ using
\eqref{eq:success}. The integers $N,D$ have
$O(s\log q)$ bits, and the two amplitudes are
$\sqrt{1-D/(4N)}$ and $\sqrt{D/(4N)}$.
The bound $p_{q,s}\ge1/2$ guarantees that both
radicands lie in $[0,1]$. Deterministic integer
arithmetic and square-root approximation compute
these coefficients to $t$ bits in time polynomial
in $s\log q+t$, so the gate has a short uniform
description.

For a fixed input, $s$ lies in $\{2,\ldots,n\}$, so there are at most
$n-1$ attenuation coefficients for a given field. They may be tabulated
before factoring or computed as each block is reached. Different
measurement branches produce different block polynomials, but no angle
depends on the unknown CRT matrix. The field size $q$ varies among inputs
of the same binary length, and the angle varies with it; this is why the
model requires uniformity in the input rather than in its length.

We use the following particular form of known-probability exact amplitude
amplification \cite[Theorem~4]{BHMT02}. Exact amplification with a known
success probability goes back to Brassard and H{\o}yer \cite{BH97}, and
Mosca and Zalka \cite{MZ03} used it to make the discrete-logarithm
algorithm exact; attenuating the marked probability to $1/4$ is the
standard way to reach certainty with a single iterate. If $A\ket0$ has marked
probability $p\ge1/4$, prepare an independent qubit in
\[
 C_p\ket0=\sqrt{1-\frac1{4p}}\,\ket0+
                 \frac1{2\sqrt p}\,\ket1.           \tag{2}\label{eq:attenuation}
\]
Mark a basis state exactly when the original predicate is true \emph{and}
the auxiliary qubit is $1$. For $A'=A\otimes C_p$ this augmented predicate
has probability $p'=1/4$. Consequently one standard Grover iterate
\[
 Q'=-A'S_0(A')^{-1}S_{\chi'},\qquad
 \chi'(\lambda,b)=\chi(\lambda)\land[b=1],           \tag{3}\label{eq:grover}
\]
maps $A'\ket0$ entirely into the augmented marked subspace, since
$\sin^2(3\arcsin\sqrt{p'})=1$. Measuring its original register gives a
solution with certainty. In particular, the reflections and marked oracle
in \eqref{eq:grover} act on \emph{both} registers; one iteration with
$-AS_0A^{-1}S_\chi$ need not be exact.

It is helpful to check the endpoint $p=1/2$, which occurs for
characteristic $2$ and $s=2$. An unmodified Grover iteration starts
at angle $\theta=\pi/4$ and ends at $3\theta=3\pi/4$; its success
probability is still $1/2$. In contrast, \eqref{eq:attenuation}
places amplitude $1/\sqrt2$ on the auxiliary state $\ket1$.
The augmented marked probability becomes $1/4$, whose angle is
$\pi/6$; a single iterate ends at $3\pi/6=\pi/2$.
The auxiliary bit therefore participates in both
the marking predicate and the reflection about the prepared state.

The preparation and its inverse are used a constant number of times:
starting from $A'\ket0$, the iterate applies one augmented marking,
one $(A')^{-1}$, a phase flip about the all-zero state, and one $A'$.
The reflection about zero is a multi-controlled phase gate on the
whole working register, including the auxiliary bit. Its Boolean-gate
cost is polynomial in the register size; it does not require an
additional algebra multiplication of the order of a full test
evaluation.

\section{Exact splitting of polynomials}\label{sec:main}
Choose a basis $v_1,\ldots,v_s$ of $B_P$ with $v_1=1$, and let
$a_\lambda=\sum_{k=1}^s\lambda_kv_k$ for $\lambda\in\F_q^s$.
The map from $\lambda$ to its CRT coordinate tuple in $\F_q^s$ is a
bijection. Indeed, the coefficient map $\F_q^s\to B_P$ is a linear
isomorphism by definition of a basis, and restriction of the CRT is
another linear isomorphism. The matrix of their composition depends
on the unknown factors, but the image of a uniform input under
\emph{any} bijection is uniform. This is why knowing $s$, rather than
knowing the coordinate maps, suffices for exact amplification.
The observation yields the useful general counting statement.

\begin{lemma}[Fiber count]\label{lem:count}
If $F\in\F_q[X]$ has image $T\subseteq\F_q$ and fibers of sizes $c_t$, then
\[
 \Pr_{\lambda\in\F_q^s}[F(a_\lambda)\text{ is constant in }B_P]
       =q^{-s}\sum_{t\in T}c_t^s.                  \tag{4}\label{eq:count}
\]
Writing $\gamma=\max_t c_t/q$, this probability is at most $\gamma^{s-1}$.
\end{lemma}
\begin{proof}
The CRT tuple of $a_\lambda$ is uniform on $\F_q^s$. For a fixed $t$,
exactly $c_t^s$ tuples lie in the $t$-fiber in every coordinate. Sum over
$t$; the bound follows from
$\sum_t(c_t/q)^s\le\gamma^{s-1}\sum_t c_t/q$.
\end{proof}

This is the count used to bound the failure probability of a random
splitting element in Berlekamp's and Cantor--Zassenhaus's algorithms
\cite{Berlekamp70,CZ81}; see \cite[Ch.~14]{vzGG13}. The formula applies to all $s\ge1$, but a split is possible only
for $s\ge2$. Its exact value, rather than just the upper bound in
Lemma~\ref{lem:count}, is needed for the rotation
\eqref{eq:attenuation}. An upper bound on the failure probability
alone would give bounded error, not certainty. Since the formula
depends on fiber sizes only, it applies again after an adaptive split:
the new block may contain a different unknown subset of factors, but
the next probability is obtained by replacing $s$ with that block's
computed basis dimension.

For odd $q$, use $F(X)=X^{(q-1)/2}$, whose fiber sizes are
$1,(q-1)/2,(q-1)/2$. For $q=2^k$, use
$F(X)=\Tr_{\F_q/\F_2}(X)=\sum_{j=0}^{k-1}X^{2^j}$,
whose two fibers both have size $q/2$.

For odd $q$, Euler's criterion gives the image
$\{0,1,-1\}$: zero has one preimage, while each nonzero
value has $(q-1)/2$ preimages. The absolute trace for
$q=2^k$ is a nonzero $\F_2$-linear functional on the
$k$-dimensional $\F_2$-space $\F_q$, hence each of its
two fibers has $2^{k-1}=q/2$ elements. Both tests are
polynomials with coefficients in $\F_q$, so evaluation on
$a_\lambda\in B_P$ commutes with all coordinate maps
$\pi_i:B_P\to\F_q$. Their images remain inside $B_P$.

\begin{lemma}[Exact success probability]\label{lem:prob}
For $s\ge2$, the probability that $w_\lambda=F(a_\lambda)$ is
nonconstant is
\[
 p_{q,s}=
 \begin{cases}
  1-\dfrac{1+2((q-1)/2)^s}{q^s},&q\text{ odd},\\[1ex]
  1-2^{1-s},&q=2^k .
 \end{cases}                                      \tag{5}\label{eq:success}
\]
It depends only on $q,s$ and satisfies $1/2\le p_{q,s}<1$.
\end{lemma}
\begin{proof}
Substitute the fiber sizes into \eqref{eq:count}. The largest relative
fiber size is less than $1/2$ for odd $q$ and equals $1/2$ in
characteristic $2$. Lemma~\ref{lem:count} gives the bound; every test
also has a nonempty constant fiber, so $p_{q,s}<1$.
\end{proof}

For illustration, when $s=2$ the characteristic-two probability is
$1/2$. In odd characteristic it is
\[
 p_{q,2}=1-\frac{1+2((q-1)/2)^2}{q^2}
          =\frac{q^2+2q-3}{2q^2}>\frac12.
\]
The probability increases with $s$: a tuple is less likely to lie
entirely in one fiber as more independent coordinates are sampled.
The input $q$ and the dimension $s$ can have many possible values,
but the numerator and denominator in \eqref{eq:success} have only
$O(s\log q)$ bits, so computing the required rotation to
prescribed precision is polynomial-time classical work.

\subsection{Why these two tests}\label{sec:choice-tests}
The test must meet three demands at once. Its evaluation on an
element of $B_P$ must take only polynomially many field operations
in $\log q$; its image must contain only a few \emph{known}
values, so that gcd extraction does not scan $\F_q$; and the
fiber sizes must be explicit enough to specify an exact
rotation. A balance condition ensures that one amplification
iteration is sufficient. The identity test $F(X)=X$ has
excellent success probability, namely $1-q^{1-s}$, but
its $q$ candidate values make extraction expensive.
Conversely a constant map has tiny image but never
separates a block. The chosen character and trace maps
meet all three demands simultaneously.

It might seem that the absolute trace to $\F_p$ could serve
for every characteristic $p$. The map itself has $p$
equal-size fibers and success probability
$1-p^{1-s}$, so its balance is good. However, its
image has $p$ elements, and trying all $p$ gcd
candidates is not polynomial in $\log q$ when $p$
grows. Coarsening its output to the predicate
\emph{trace equals zero} gives a two-valued test but
with relative fiber sizes $1/p$ and $(p-1)/p$.
For fixed $s$, the latter fiber makes the probability
of a useful nonconstant result tend to zero as
$p\to\infty$. This is why the trace test is used
when $p=2$, while odd characteristic uses
the quadratic-character partition.

For fixed small $p$ the situation is different:
exhausting the $p$ trace values is affordable
classically. More strongly, one can avoid random
sampling altogether. Regard $B_P\cong\F_q^s$
as an $\F_p$-space and apply the absolute trace in
each coordinate. The resulting map onto
$\F_p^s$ is $\F_p$-linear and surjective.
Its constant tuples form the one-dimensional
diagonal subspace. If $s\ge2$, the images of all
vectors in an $\F_p$-basis of $B_P$ cannot lie
inside that diagonal. Sweeping that basis and
testing the $p$ possible values therefore finds
a split deterministically. Our characteristic-two
result is accordingly about a uniform exact
quantum primitive, not about resolving an open
classical complexity question in small
characteristic.

\subsection{The oracle and a split}\label{sec:oracle}
Compute $a_\lambda$ and then $w_\lambda$ reversibly. Binary exponentiation
for odd $q$, or the sequence of squares and sums for $q=2^k$, uses
$O(\log q)$ multiplications in $R_P$. Schoolbook multiplication modulo
$g_P$ takes $O(n_P^2)$ $\F_q$-operations and can be made cleanly
reversible by retaining the inputs, computing the output in a fresh
register, and uncomputing the intermediate product and quotient.
Thus the test costs $O(n_P^2\log q)$ field operations.
Appendix~\ref{app:reversible} gives the clean multiplication and evaluation
circuits, including their scratch registers and uncomputation.

Represent $w_\lambda$ by its degree-$<n_P$ polynomial. It is constant
exactly when all nonconstant coefficients vanish. Computing their
Boolean OR reversibly gives a bit-flip predicate oracle
$U_\chi:\ket\lambda\ket z\mapsto\ket\lambda\ket{z\oplus\chi(\lambda)}$.
Compute, apply a controlled phase conditional on this predicate
\emph{and} the auxiliary bit of \eqref{eq:attenuation}, and uncompute.
This is $S_{\chi'}$ in \eqref{eq:grover}. All arithmetic scratch
registers are returned to zero. Preparing $A_P$ and its inverse has
no greater asymptotic arithmetic cost than evaluating the test.

The predicate tests only constancy; polynomial gcds are computed
classically after measurement. Details of the reversible comparison are
in Appendix~\ref{app:constant}.

\begin{lemma}[One exact splitting round]\label{lem:split}
For every block with $s\ge2$, one augmented Grover iteration produces a
$\lambda$ with nonconstant $w_\lambda$ with probability exactly one.
A nontrivial proper divisor of $g_P$ is then recovered classically using
at most three polynomial gcd computations (two in characteristic $2$).
The quantum round costs $O(n_P^2\log q)$ $\F_q$-operations.
\end{lemma}
\begin{proof}
Apply \eqref{eq:attenuation} with $p=p_{q,s}$ and use
\eqref{eq:grover}. Measurement necessarily yields $\chi(\lambda)=1$.
For the resulting concrete polynomial $w_\lambda$, try
$\gcd(g_P,w_\lambda-t)$ for $t\in\{0,1,-1\}$ in odd characteristic, or
$t\in\{0,1\}$ in characteristic $2$. At least two CRT coordinates have
different values in that set. Hence some gcd is a proper nontrivial
factor, and its quotient supplies the other child. The arithmetic
bound includes computation and uncomputation of the controlled predicate.
\end{proof}

The augmented construction also preserves a useful distributional
property. In the prepared state, each $\lambda\in\F_q^s$
has the same amplitude, and the auxiliary amplitude on
$\ket1$ is independent of $\lambda$. The normalized
projection onto the augmented marked subspace is
therefore a uniform superposition over the $\lambda$
with $\chi(\lambda)=1$, with the auxiliary bit fixed
to $1$ and the deterministically computed work
registers attached. Standard Grover rotation changes
the weights of the marked and unmarked subspaces,
but not relative amplitudes \emph{within} the marked
subspace. Once the unmarked component vanishes,
measurement samples a uniformly random useful
$\lambda$. Exact amplification thus removes
failure without selecting a special CRT basis
vector or imposing a preferred split shape.

The gcd conclusion deserves a small distinction: a nonconstant
$w_\lambda$ need not make \emph{every} candidate $t$ useful.
A value absent from its CRT tuple gives gcd $1$, while a value shared
by all coordinates gives gcd $g_P$; nonconstancy ensures that at least
two values occur, so any occurring value has a nonempty proper
fiber. Trying the constant-size image of the test finds one.
For example, over $\F_5$ a hypothetical CRT tuple
$(a_1,a_2,a_3)=(0,1,2)$ is mapped by $a\mapsto a^2$ to
$(0,1,-1)$. The gcds at $0,1,-1$ isolate the three corresponding
factors, even though the algorithm never constructs that tuple.

For a child divisor $g_{P'}\mid g_P$, reduction
$R_P\to R_{P'}$ maps $B_P$ \emph{onto} $B_{P'}$: under (1) it forgets
coordinates. Reducing the parent basis and selecting an independent
subset by Gaussian elimination therefore gives a child basis, with
$1$ retained as its first vector.

\begin{lemma}[Basis pushdown]\label{lem:pushdown}
Let $g'\mid g$ be monic and squarefree and put
$B_g=\Fix(a\mapsto a^q,\F_q[x]/(g))$, with analogous
notation for $g'$. The reduction homomorphism
$\rho:\F_q[x]/(g)\to\F_q[x]/(g')$ satisfies
$\rho(B_g)=B_{g'}$. If $1,v_2,\ldots,v_s$
is a basis of $B_g$, a basis of $B_{g'}$ can
be selected among $1,\rho(v_2),\ldots,\rho(v_s)$.
\end{lemma}
\begin{proof}
Reduction commutes with the $q$-power map, so
$\rho(B_g)\subseteq B_{g'}$. Write $g$ as the product
of its distinct monic irreducible factors and write
$g'$ as the product of the subset dividing $g'$.
Under the two CRT identifications, $\rho$ restricts
to the coordinate projection from one product of
copies of $\F_q$ to the smaller product. This
projection is surjective and maps the tuple
$(1,\ldots,1)$ to the corresponding tuple for $g'$.
Thus the reduced vectors span, and ordinary
elimination can retain $1$ as the first basis vector.
\end{proof}

The lemma remains valid if the new divisor is a union
of several irreducible factors rather than just one.
It is what permits the splitting process to carry
one fixed-point basis down the recursion instead of
solving a fresh Frobenius null-space system for
every node. The proof is existential through the
CRT, whereas the implementation uses only
polynomial reduction and elimination on explicit
coefficient vectors.

The recursive algorithm can be specified without naming the hidden
sets $P$ in its input:
\begin{enumerate}
\item Compute the Berlekamp basis of the squarefree input $f$ and place
      $(f,v_1,\ldots,v_r)$ on a worklist.
\item Remove a pair $(g,v_1,\ldots,v_s)$. If $s=1$, output $g$.
      Otherwise compute \eqref{eq:success}, prepare
      \eqref{eq:attenuation}, apply the augmented iterate
      \eqref{eq:grover}, and measure $\lambda$.
\item Recompute $w_\lambda$ classically and test the at most three
      candidate gcds. Divide $g$ by a nontrivial gcd $h$ to obtain
      two smaller squarefree divisors $h$ and $g/h$.
\item Reduce the parent basis modulo each divisor, retain a basis of
      each image beginning with $1$, and place the two pairs on the
      worklist. Continue until it is empty.
\end{enumerate}
The list contains only explicit polynomials and explicit basis
vectors. Its termination condition $s=1$ is exact: for a squarefree
polynomial, the dimension of the fixed-point algebra equals its
number of irreducible factors, regardless of factor degrees or of the
test values seen in earlier rounds.

For odd $q$, the three candidate gcds in step~3 may expose three
nonempty fibers at once. Using all of them can only shorten the
splitting tree; the multiway variant takes at most $r-1$ rounds, as in
Section~\ref{sec:generalization}. We state the binary version so that
the number of rounds is exactly $r-1$ for every measurement outcome.

\begin{theorem}[Exact Berlekamp factorization]\label{thm:main}
Let $f$ be monic and squarefree of degree $n\ge1$, with $r$ irreducible
factors. In the circuit model of Section~\ref{sec:model}, its complete
factorization is produced with probability exactly one by precisely
$r-1$ splitting rounds. Their total quantum arithmetic cost is
\[
 O\!\left(\sum_{P\ {\rm internal}} n_P^2\log q\right)
       \subseteq O(n^3\log q)                      \tag{6}\label{eq:quantumcost}
\]
$\F_q$-operations. The straightforward deterministic preprocessing and
classical bookkeeping take at most $O(n^4+n^3\log q)$ $\F_q$-operations
under schoolbook arithmetic.
\end{theorem}
\begin{proof}
Compute the initial fixed-point basis by linear algebra. On each block of
dimension $s\ge2$, Lemma~\ref{lem:split} gives a genuine split and the
surjectivity above supplies child bases. A binary splitting tree with $r$
leaves has $r-1$ internal nodes. Each $n_P\le n$ and $r\le n$, giving
\eqref{eq:quantumcost}. A conservative bound for the initial Frobenius
matrix and null space is $O(n^3\log q)$. For a child, row reduction of
the $n_{P'}\times s$ matrix of reduced parent vectors costs
$O(n_{P'}s^2)\subseteq O(n^3)$; with at most $2(r-1)$ children this is
$O(n^4)$ overall. Gcd extraction and measured-element recomputation
cost no more than the stated polynomial bound.
\end{proof}

The initial Frobenius computation is a one-off deterministic step:
compute $x^q\bmod f$, build the matrix on $1,x,\ldots,x^{n-1}$ by modular
products, and row reduce. This fits the $O(n^3\log q)$ allowance.
Appendix~\ref{app:costs} separates this cost from subsequent bookkeeping;
the latter includes classical recomputation of each measured test.
The sum in \eqref{eq:quantumcost} depends on tree shape: if all factors
are linear and each split removes one, $\sum n_P^2=\Theta(n^3)$, whereas
balanced splits give $O(n^2)$. The algorithm does not control which useful
$\lambda$ is measured, so the theorem states the worst case; the number of
rounds, $r-1$, does not depend on the shape. Appendix~\ref{app:adaptive}
explains conditional exactness at adaptive nodes and the distinction
between hybrid and deferred-measurement costs.

\subsection{Examples and boundary cases}
Take $q=5$ and, only for illustration, suppose
$f=x(x-1)(x-2)$. All three irreducible factors are linear,
so $B=R$ and the residue class $a=x$ is available as a
possible test input. Its CRT tuple is $(0,1,2)$, and the
quadratic-character polynomial gives
$w=a^2$ with tuple $(0,1,-1)$. Therefore
\[
 \gcd(f,w)=x,\quad
 \gcd(f,w-1)=x-1,\quad
 \gcd(f,w+1)=x-2.
\]
The algorithm does not know these factors in advance or
force its measured test input to equal $x$; this example
illustrates why a measured \emph{nonconstant} test suffices
for deterministic extraction. The exact probability for
an arbitrary uniformly sampled element of this
three-dimensional fixed-point algebra is
$p_{5,3}=1-(1+2\cdot 2^3)/5^3=108/125$.

For a characteristic-two illustration, let
$\F_4=\F_2(\omega)$ with $\omega^2+\omega+1=0$.
The absolute trace $t\mapsto t+t^2$ has fibers
$\{0,1\}$ and $\{\omega,\omega^2\}$. If a block
contains three linear factors and a test input happens
to have CRT tuple $(0,1,\omega)$, its trace tuple is
$(0,0,1)$; the gcd at $1$ separates the third factor.
The probability of obtaining a useful trace tuple from
a \emph{uniform} element is $p_{4,3}=3/4$.
Notice that $\Tr_{\F_4/\F_2}(1)=0$; the trace map is
nonzero and surjective even though it need not send
$1$ to $1$.

If $s=1$, every element of $B_P$ is constant and no
splitting circuit should be invoked. The basis dimension
already certifies irreducibility of $g_P$. At the other
extreme, a block can split into a singleton and a large
child repeatedly; correctness is unchanged, and
Theorem~\ref{thm:main} uses a tree-independent upper
bound for this worst shape. The auxiliary rotation is
never singular, since \eqref{eq:success} gives
$1/2\le p_{q,s}<1$ whenever the circuit is invoked.

\section{Commutative algebras}\label{sec:generalization}
Let $A$ be a commutative unital $\F_q$-algebra of dimension $n\ge1$,
given by structure constants. Its identity is supplied or recovered by
solving the linear equations $1_Au_j=u_j$ on an input basis. Primitive
idempotent computation for commutative Artinian algebras over finite fields
is a classical problem; see, for example, \cite{Davis94}.

\subsection{The fixed-point algebra in a generic presentation}
Write $J=\Rad(A)$ and
\[
 A/J\cong\prod_{i=1}^r\F_{q^{d_i}}.
\]
Here $r$ is the number of local factors of $A$, equivalently the number
of primitive idempotents. The radical is used only in the analysis.

\begin{lemma}[Fixed points modulo the radical]\label{lem:nilpotents}
For any finite-dimensional commutative unital $\F_q$-algebra $A$,
reduction modulo $J=\Rad(A)$ induces an $\F_q$-algebra isomorphism
\[
 B=\Fix(a\mapsto a^q,A)
   \ \xrightarrow{\ \sim\ }\
 \Fix(a\mapsto a^q,A/J)\cong\F_q^r.
\]
Moreover, $A$ and $B$ have the same idempotents and the same primitive
idempotents.
\end{lemma}
\begin{proof}
Since $A$ is finite dimensional, $J$ is nilpotent. Choose $N\ge1$ with
$J^N=0$ and an integer $\ell\ge0$ with $q^\ell\ge N$. If $b\in B\cap J$,
then $b=b^{q^\ell}=0$, proving injectivity. For a fixed point $\bar a\in A/J$,
choose any lift $a\in A$ and put $b=a^{q^\ell}$. Commutativity gives
\[
 b^q-b=(a^q-a)^{q^\ell}=0,
 \qquad \bar b=\bar a^{q^\ell}=\bar a,
\]
because $a^q-a\in J$. Thus reduction is surjective. The fixed points form
an $\F_q$-subalgebra, and the fixed field of the $q$-power map in every
residue field $\F_{q^{d_i}}$ is $\F_q$. Finally every idempotent $e$ satisfies
$e^q=e$, so all idempotents of $A$ lie in $B$. Their possible orthogonal
decompositions are therefore the same in both algebras, which also
identifies the primitive idempotents.
\end{proof}

With an input basis $u_1,\ldots,u_n$, the structure constants
$c_{ij}^k$ specify $u_iu_j=\sum_k c_{ij}^ku_k$. A direct product costs
$O(n^3)$ field operations. The $q$-power map is $\F_q$-linear;
its matrix has columns $u_j^q$, computed by binary powering. Computing
the kernel of this matrix minus the identity yields $B$ and $r$ directly.

A current block is $A_P=e_PA$, with its own identity $e_P$. All its
operations can be carried out in ambient coordinates. Polynomial
evaluation in $A_P$ interprets a scalar constant $c$ as $ce_P$.

\begin{lemma}[Fixed points on an idempotent block]\label{lem:block}
Let $e\in A$ be a sum of $s$ primitive idempotents of a commutative
unital $\F_q$-algebra. Then
\[
 \Fix(a\mapsto a^q,eA)=eB\cong\F_q^s,
 \qquad
 \F_qe\longleftrightarrow\{(t,\ldots,t):t\in\F_q\}.
\]
A uniform linear combination of a basis of $eB$ has independent uniform
coordinates in this product.
\end{lemma}
\begin{proof}
For $b\in B$, $(eb)^q=eb$. Conversely a fixed point in $eA$ belongs
to $B$ and equals its product with $e$. Under the isomorphism in
Lemma~\ref{lem:nilpotents}, multiplication by $e$ is projection onto its
$s$ selected coordinates. Its image has identity $(1,\ldots,1)$, and
any basis coefficient map is a bijection onto that image.
\end{proof}

Consequently, multiplying a parent basis by a child idempotent and
selecting independent vectors gives a child basis beginning with that
idempotent. No new Frobenius kernel is needed. The relevant constant
subspace is $\F_qe_P$, and the fiber count of Lemma~\ref{lem:count}
applies unchanged.

For example, let
$A=\F_q[\varepsilon]/(\varepsilon^2)\times
\F_q[\eta]/(\eta^3)$. Then $B=\F_q\times\F_q$ and its primitive
idempotents are $(1,0)$ and $(0,1)$. Splitting $B$ separates the two
local factors of $A$. In general the output is $A=\prod_i e_iA$ with local factors, not necessarily fields;
the factors are fields when $A$ is reduced.

\subsection{Effective tests and extraction}
The two tests from Section~\ref{sec:main} suffice and have image size
at most three. It is useful to state exactly what a more general test
would require.

\begin{definition}[Effective test map]\label{def:testmap}
An effective test map for $\F_q$ is a polynomial $F\in\F_q[X]$
specified together with (i) a reversible straight-line evaluation
program of polynomial length in $\log q$, (ii) a list
$T=F(\F_q)$ of polynomial size in $\log q$, and (iii) the exact
fiber cardinalities $c_t=|F^{-1}(t)|$, computable deterministically
in polynomial time in $\log q$. Require additionally
$\max_{t\in T}c_t\le q/2$. For the $O(\log q)$ bounds below, assume
the evaluation program has $O(\log q)$ algebra operations.
\end{definition}

The fiber bound ensures a useful success probability, while the exact
cardinalities specify the rotation and the image list permits extraction.
These data must be supplied or efficiently computable: a short evaluation
circuit alone need not reveal its image or fiber sizes.

Besides the two maps used here, the definition admits variants.
If $d\mid(q-1)$ and $d\ge2$, the power map
$X^{(q-1)/d}$ has image $\{0\}\cup\mu_d$ and fiber sizes
$1,(q-1)/d,\ldots,(q-1)/d$. It is effective whenever $d$
is small and the $d$ roots of unity are supplied by an
efficient deterministic enumeration. For odd $q$, taking $d=2$
gives $\{0,1,-1\}$ without constructing a distinguished
quadratic non-residue. For $q=2^k$ the absolute trace gives the
two values $0,1$ with balanced fibers, including when $k>1$.
More generally, for $q=p^k$ the absolute trace
$\Tr_{\F_q/\F_p}(X)=\sum_{j<k}X^{p^j}$ has $p$ fibers of size $q/p$, is
evaluated with $O(k\log p)=O(\log q)$ multiplications, and gives
$p_{F,s}=1-p^{1-s}$. It is an effective test whenever $p$ is bounded by a
polynomial in $\log q$; Section~\ref{sec:choice-tests} explains why it
fails to be one when $p$ is large.
These examples show why the field characteristic enters the
concrete choice of test even though the counting lemma is
characteristic-free.

The quadratic-character and characteristic-two trace maps meet this
definition with $|T|\le3$. The known cardinalities, not merely the
upper bound $c_t\le q/2$, are needed to compute the exact rotation.
For any effective test, the probability of
$F(a_\lambda)\notin\F_qe_P$ is
\[
 p_{F,s}=1-q^{-s}\sum_{t\in T}c_t^s\ge\frac12.    \tag{7}\label{eq:generalprob}
\]
Indeed, Lemma~\ref{lem:block} again makes the coordinates of
$a_\lambda$ independent and uniform. The augmented oracle of \eqref{eq:grover}, now testing
membership in $\F_qe_P$, produces a nonconstant $w=F(a_\lambda)$
with certainty. Choose any nonzero coordinate $(e_P)_j$ of the known
idempotent: $w\in\F_qe_P$ exactly when
$w=(w_j/(e_P)_j)e_P$. This comparison is reversible in $O(n)$
field operations.

\begin{lemma}[Idempotent extraction]\label{lem:idempotents}
For such a $w$ and each $t\in T$, set
\[
 e_t=\prod_{\substack{u\in T\\u\ne t}}
                 \frac{w-ue_P}{t-u}\in B_P.       \tag{8}\label{eq:lagrange}
\]
The nonzero $e_t$ are orthogonal idempotents summing to $e_P$, and at
least two are nonzero. All $e_t$ can be computed using $O(|T|)$ algebra
multiplications and polynomially many scalar operations.
\end{lemma}
\begin{proof}
In each coordinate of $B_P\cong\F_q^s$ (Lemma~\ref{lem:block}),
\eqref{eq:lagrange} is one if that coordinate of $w$ equals
$t$, and zero otherwise. Nonconstancy gives two nonempty fibers.
Prefix and suffix products of the factors $w-ue_P$ compute every
numerator in $O(|T|)$ algebra multiplications in total; all
denominators are nonzero field elements.
\end{proof}

The known list $T$ may contain values absent from the current block;
their idempotents are zero and are discarded. Thus extraction does not
require factoring a minimal polynomial. For $m=|T|$, all scalar
denominators can be computed once in $O(m^2)$ field operations and
inverted using $O(m)$ inversions; prefix and suffix products give the
numerators in $O(m)$ algebra multiplications per round.

For example, take $A=\F_q\times\F_q\times\F_{q^2}$.
Its Frobenius-fixed algebra is $\F_q^3$ even though
$\dim_{\F_q}A=4$. If a measured test element has coordinates
$(0,1,1)$, \eqref{eq:lagrange} yields the two nonzero
idempotents $(1,0,0)$ and $(0,1,1)$. Multiplying the parent
basis by the second idempotent gives a two-dimensional fixed
subalgebra, so a further round separates its remaining two
primitive idempotents. The number of internal nodes can be
\emph{less} than $r-1$ when one test produces three nonempty
fibers at once.

Multiplication by $e_t$ maps the parent fixed-point algebra onto the
child one. Hence the images of its basis span each child and Gaussian
elimination supplies a basis beginning with $e_t$.

\begin{theorem}[Exact commutative splitting]\label{thm:commutative}
Let $A$ be an $n$-dimensional commutative unital $\F_q$-algebra,
$n\ge1$, given by structure constants, and let
$r=\dim_{\F_q}\Fix(a\mapsto a^q,A)$. Using either concrete test above,
its primitive idempotents and the local factors $e_iA$ can be computed
with certainty in at most $r-1$ splitting rounds and
$O(n^4\log q)$ quantum $\F_q$-operations. Deterministic classical work
also takes $O(n^4\log q)$ $\F_q$-operations under naive arithmetic.
\end{theorem}
\begin{proof}
Lemmas~\ref{lem:nilpotents} and \ref{lem:block} identify each block's
fixed-point algebra with $\F_q^s$. Equation~\eqref{eq:generalprob} and
the augmented circuit give a nonconstant test element with certainty;
Lemma~\ref{lem:idempotents} splits the block into at least two children.
There are at most $r-1\le n-1$ internal nodes. Generic multiplication
costs $O(n^3)$ field operations, so each quantum evaluation costs
$O(n^3\log q)$. The initial Frobenius matrix costs $O(n^4\log q)$
classical operations. After each measurement, recomputing the test
costs $O(n^3\log q)$ and extracting its at most three idempotents costs
$O(n^3)$. For each child, constructing its multiplication map, applying
it to a parent basis, and eliminating dependencies each cost $O(n^3)$.
A tree with $r$ leaves and at least two children per internal node has
at most $2r-2$ edges, so all child updates cost $O(rn^3)$. Finally a
basis of each local factor $e_iA$ is obtained from the image of its
multiplication map in $O(n^3)$ operations. These costs give the bounds.
\end{proof}

\begin{samepage}
For a general effective test of evaluation length $L\ge1$ and image
size $m$, the same argument gives $O(rn^3L)$ quantum field operations
and
\[
 O\bigl(n^4\log q+rn^3L+rmn^3+m^2\bigr)
\]
classical field operations, together with $O(m)$ inversions for the
scalar interpolation denominators and the work to specify the test.
The $rn^3L$ term includes recomputing each measured test value; the
$m^2$ term prepares the denominators once for the fixed list $T$.
For the concrete tests, $L=O(\log q)$ and $m\le3$.
\end{samepage}

\begin{remark}[Comparison with the cutting procedure]\label{rem:cutting}
For reduced $A$, Theorem~\ref{thm:commutative} can also be reached by a
shorter route. The cutting procedure of Friedl and R\'onyai
\cite{FR85,Ronyai90} is deterministic apart from calls to a
polynomial-factoring oracle, so answering those calls with
Theorem~\ref{thm:main} already gives an exact algorithm. The difference
lies in the fields over which the calls are made. That procedure walks a
basis $a_1,\ldots,a_n$ of $A$, keeping $F_i=\F_q(a_1,\ldots,a_i)$ a field
of degree $e_i\le n$ over $\F_q$; at step $i$ it computes the minimal
polynomial of $a_{i+1}$ over $F_i$ and factors it over $F_i$. A call over
$F_i$ handled by Theorem~\ref{thm:main} costs $O(n^3e_i\log q)$
operations in $F_i$, each of which is $\tilde O(e_i)$ operations in
$\F_q$, for $\tilde O(n^3e_i^2\log q)$ operations in $\F_q$ and
$\tilde O(n^5\log q)$ in the worst case. With up to $n$ steps this gives
$\tilde O(n^6\log q)$, against $O(n^4\log q)$ in
Theorem~\ref{thm:commutative}. Neither bound is optimized, and we do not
claim that either exponent is tight. The structural difference does not
depend on optimization: the direct construction never builds a
generating element, never leaves $\F_q$, and uses one splitting round
per node of the decomposition tree rather than one factoring call per
basis vector.
\end{remark}

\section{Associative algebras and scope}\label{sec:noncommutative}
Let $A$ now be a unital associative $\F_q$-algebra of dimension
$n\ge1$, given by structure constants. Its identity can again be
supplied or recovered by linear algebra. R\'onyai's deterministic
algorithm computes $\Rad(A)$ in polynomial time without a factoring
oracle \cite{Ronyai90}; see also \cite{CIW97}. In small characteristic
the trace pairing alone does not characterize the radical, since it can
be degenerate even for a semisimple algebra when the characteristic
divides a relevant matrix dimension.

Set $S=A/\Rad(A)$. Once a radical basis is known, linear algebra
supplies a quotient basis and its induced structure constants.
Its center is found by solving $za=az$ on a basis of $S$:
each commutator condition is linear in the coefficients of $z$.
Multiplying the resulting center basis elements and expressing
their products in that basis gives the center's structure
constants. These steps are deterministic polynomial-time
computations. By
Artin--Wedderburn,
\[
 S\cong\prod_i M_{k_i}(\F_{q^{d_i}}),\qquad
 Z(S)\cong\prod_i\F_{q^{d_i}}.
\]
Primitive idempotents $e_i$ of $Z(S)$ are central primitive
idempotents of $S$, and the minimal two-sided ideals are $e_iS$.
The center is commutative, so Theorem~\ref{thm:commutative}
applies; since $Z(S)$ is reduced, its local factors are the fields
$\F_{q^{d_i}}$. Isomorphic simple summands
still have distinct primitive central idempotents, so separating
them requires no isomorphism test.

This is a decomposition into \emph{simple two-sided ideals}, not
an explicit identification of each ideal with matrices. For
example, if $S\cong M_2(\F_q)\times M_3(\F_{q^2})$, then
$Z(S)\cong\F_q\times\F_{q^2}$ has dimension three over $\F_q$
but only two primitive idempotents. Splitting the center gives
those two idempotents and thus the two ideals; producing
matrix units within either ideal is an additional task.
The dimension of the center must not be confused with its
number of field components.

\begin{corollary}[Exact Wedderburn--Artin decomposition]\label{cor:wedderburn}
The radical and the decomposition of $A/\Rad(A)$ into minimal
two-sided ideals are computable with certainty using
$O(n^4\log q)$ quantum $\F_q$-operations and polynomially many
deterministic classical field operations, in the circuit model of
Section~\ref{sec:model}.
\end{corollary}
\begin{proof}
Compute the radical and the center deterministically, apply
Theorem~\ref{thm:commutative} to the center, whose dimension is at most
$n$, and form the ideals $e_iS$. Each possible exact output
idempotent family yields the asserted decomposition.
\end{proof}

The $O(n^4\log q)$ bound is for quantum arithmetic inside
center splitting. Constructing the quotient and center is
included among the polynomial classical operations. In
particular, an arbitrary presentation of $A$ need not support
the $O(n^2)$ modular multiplication available for a polynomial
quotient.

Producing a zero divisor inside a simple component, or an explicit
isomorphism with $M_k(\F_{q^d})$, asks for more than central idempotents:
it requires a concrete center field, a system of matrix units, and
compatible coordinates. R\'onyai's constructions for these tasks
\cite{Ronyai90} are deterministic reductions to polynomial factoring, so
answering their calls with Theorem~\ref{thm:main} makes them exact as
well. Their calls may, however, be over extension fields $\F_{q^d}$, and a
call over $\F_{q^d}$ invokes Theorem~\ref{thm:main} with field size $q^d$,
with each $\F_{q^d}$-operation then charged over $\F_q$ (compare
Remark~\ref{rem:cutting}). We therefore assert no $O(n^4\log q)$ bound for
these tasks; Corollary~\ref{cor:wedderburn} covers only the radical and
the minimal two-sided ideals.

\section{Comparison and limitations}\label{sec:comparison}
\subsection{Classical factorization regimes}
The role of characteristic is easy to obscure if all finite fields
are discussed together. The original Berlekamp trace method
\cite{Berlekamp67} can sweep an $\F_p$-basis and exhaust the
$p$ possible trace values; for fixed $p$ this gives a
deterministic algorithm polynomial in $n,\log q$. Thus the
characteristic-two case of our construction has no new
deterministic-versus-randomized complexity consequence.
It gives a uniform exact-circuit treatment and uses one
amplified test rather than a deterministic sweep. Whether that
reduces concrete resource use depends on the reversible circuit
and the field representation, and is not asserted by the
asymptotic bounds above.

For large odd characteristic, random choices of splitting
elements are efficient classically, whereas deterministic
factorization polynomial in $n,\log q$ has long been open in general
\cite{Gao01}; for prime fields, see the unrefereed claim
\cite{OpenAI26} discussed in Section~\ref{sec:intro}. There are deterministic algorithms polynomial in
$n$, $\log q$ and $\sqrt p$, where $p$ is the characteristic
\cite{Shoup90}, and conditional or structured derandomizations
\cite{Ronyai88,Evdokimov94,IKRS12,Narayanan16,Guo20}.
Our construction targets exactly the random splitting choice:
the quadratic-character map evaluates without first finding
a quadratic non-residue, and its three fiber sizes yield the
success probability directly. It still invokes deterministic
linear algebra to find the Berlekamp algebra, and it does not
claim a better asymptotic running time than classical randomized
factorization.

There are two organizations of the factoring task worth keeping
apart. A Cantor--Zassenhaus-style pipeline first removes repeated
factors, groups irreducibles by degree (distinct-degree
factorization), and then splits each group (equal-degree
factorization) \cite{CZ81}; fast randomized algorithms such as
\cite{vzGS92} follow this pipeline, and \cite{GNU16} identifies
algebraic problems equivalent to beating the exponent $3/2$ in $n$.
Berlekamp's fixed-point approach
can separate factors of different degrees in one splitting
tree; it does not require a distinct-degree stage.
Our test is therefore a \emph{Berlekamp splitting} operation,
even when it is compared to the random choice in an
equal-degree factorization algorithm.

\subsection{Quantum algorithms addressing different steps}
Doliskani's algorithm \cite{Doliskani2018} uses quantum order
finding to accelerate distinct-degree factorization and
retains a randomized splitting stage. Its performance is
best compared with other algorithms for that pipeline, while
the present result concerns zero-error splitting. These
directions could be combined in a hybrid factorization
algorithm, but the Berlekamp recursion analyzed here does
not benefit automatically from a separately computed
distinct-degree partition. In particular, we do not
transfer Doliskani's bit-complexity bound to our
field-operation model.

Draper \cite{Draper21} gives an exact quantum construction
of a quadratic non-residue over a prime field, so the absence
of a non-residue subroutine here is a simplification of the
pipeline, not a claim that such an object is quantumly
inaccessible. Draper also observes that the same device
samples exactly uniformly from any set of known cardinality
equipped with an efficient membership test. Applied to the
set of useful $\lambda$, of cardinality $p_{q,s}q^s$, this
gives Lemma~\ref{lem:split} almost directly, and we do not
claim the amplification step as new. What the device needs
is the cardinality. Here the set is defined by the unknown
CRT-coordinate \emph{tuple}; Lemmas~\ref{lem:count}
and~\ref{lem:prob} show that its size depends on $q$ and $s$
alone, and Lemmas~\ref{lem:nilpotents} and~\ref{lem:block}
show that this persists for arbitrary commutative algebras
and their idempotent blocks. Earlier
applications of exact amplification to hidden-subgroup
problems \cite{Imran-Ivanyos2022,Imran-Ivanyos2024}
likewise motivate the exactness model but do not supply
the algebraic split or the fiber count used here.

\subsection{Where the complexity statements stop}
The arithmetic bounds suppress the cost of realizing an
$\F_q$ operation by elementary gates, and the exact-gate
model permits rotations that would have to be approximated
on a fixed universal gate set. Consequently exact success
is a mathematical circuit-model claim, not a hardware
reliability claim. On a fixed finite gate set, one may
approximate the rotations and obtain a bounded-error
version, but the exact theorem would no longer follow
without an additional exact-synthesis argument.

The quantum and classical exponents count different
resources. In Theorem~\ref{thm:main}, the $O(n^3\log q)$ term
counts coherent field operations in the splitting rounds;
the straightforward basis reductions and initial
Frobenius computation have their own deterministic
upper bound. In Theorem~\ref{thm:commutative},
structure-constant multiplication is more expensive
than modular polynomial multiplication. Optimizing
these arithmetic bounds and proving tighter amortized
basis-update costs are separate algorithmic questions.

\section{Conclusion}
The classical fiber count \eqref{eq:count} depends only on the field size
and the block dimension, so a randomized algebra-splitting test becomes
amenable to exact amplitude amplification at every node of the recursion,
even though its product coordinates are hidden. This moves Berlekamp
splitting from zero-error expected polynomial time to exact quantum
polynomial time in $\mathrm{EQP}^u_{\mathbb C}$. In our attenuation
construction, the auxiliary qubit sets the marked probability to $1/4$,
enabling a single iteration with standard phase reflections; this is a
feature of the chosen construction, not a necessity for every
exact-amplification method. The resulting splitter factors squarefree
polynomials and extracts primitive idempotents of arbitrary commutative
unital algebras directly over $\F_q$. Applying it to the center of a
semisimple quotient gives the decomposition into simple ideals. The
guarantees use exact input-dependent rotations and separate quantum
arithmetic from classical preprocessing and basis maintenance.

\appendix
\section{Circuit and cost details}\label{app:implementation}
This appendix spells out the operations implicit in the asymptotic
field-arithmetic statements. It distinguishes a reversible
evaluation of the splitting predicate from a classical ability
to evaluate the predicate after a measurement.

\subsection{Exact uniform state preparation}\label{app:preparation}
Let $\ell=\lceil\log_2q\rceil$ and identify each valid
register value with an integer in $[0,q)$. For a
binary prefix $v$ of length $j$, write $N(v)$ for
the number of valid $\ell$-bit strings extending it.
Starting at the empty prefix with $N=q$, rotate
the next qubit using the amplitudes
\[
 \sqrt{\frac{N(v0)}{N(v)}}\ket0+
 \sqrt{\frac{N(v1)}{N(v)}}\ket1
\]
whenever $N(v)>0$, controlled by the preceding
prefix. Induction on $j$ gives amplitude
$\sqrt{N(v)/q}$ at prefix $v$: it holds at the
root, and multiplying by the displayed amplitude
gives the child claim. At a leaf, $N(v)$ is either
$1$ or $0$. Thus every valid string has amplitude
$1/\sqrt q$, while invalid strings have amplitude
zero.

There are at most $\ell$ boundary prefixes at
which the interval $[0,q)$ partially intersects a
binary subtree. On fully included subtrees,
subsequent bits can be prepared by Hadamards; on
excluded subtrees no gate is needed. Therefore
the construction uses polynomially many controlled
rotations and reversible comparisons with $q$.
All ratios have integer numerator and denominator
bounded by $q$, so the required gate coefficients
meet the computability convention of
Section~\ref{sec:model}. Preparing $s$ independent
field elements tensors $s$ copies of this circuit.
For $k>1$, the coefficient basis of the presentation
$\F_p[T]/(h)$ identifies $0,\ldots,q-1$ with coefficient
tuples; a reversible encoding change then feeds the
field-arithmetic registers without altering their
uniform distribution.

\subsection{Clean reversible evaluation}\label{app:reversible}
Fix a squarefree block polynomial $g$ of degree $d$.
A schoolbook circuit first forms the unreduced product of two
degree-$<d$ representatives by $O(d^2)$ products and sums
of coefficients. Division by the monic $g$ is deterministic
and has a fixed upper bound on its number of coefficient steps.
It produces a quotient and a remainder; retaining both is
necessary because the remainder alone loses information.
To implement the clean map
\[
 M_g:\ket u\ket v\ket z\longmapsto
       \ket u\ket v\ket{z+uv\bmod g},
\]
compute the product and division into separate registers,
add the remainder into $z$, and run the arithmetic backward.
Since $u,v$ remain available throughout, this gives a
permutation of computational basis states and hence a unitary.
The operation count is $O(d^2)$ field operations, up to a
constant factor for uncomputation. An exponentiation ladder
stores its intermediate powers until the final output is
copied, then uncomputes them in reverse; in characteristic
$2$, the sequence $a,a^2,\ldots,a^{2^{k-1}}$ is accumulated
by additions before uncomputation. Either ladder takes
$O(d^2\log q)$ field operations.

Let $E_g$ be the clean unitary taking $\ket\lambda\ket0$
to $\ket\lambda\ket{F(a_\lambda)}$, with its temporary
arithmetic registers reset. Let $D_g$ compute the Boolean
flag that the representative of $F(a_\lambda)$ has at
least one nonzero coefficient outside its constant term.
A clean augmented phase oracle computes $E_g$, computes
$D_g$, applies a $Z$ phase controlled on both this flag
and the qubit $b=1$, then applies $D_g^{-1}$ and
$E_g^{-1}$. All registers besides $\lambda,b$ return
to their input values, as the reflection about the
prepared state requires.

The work register for $a_\lambda$ holds $d$ field elements.
A straightforward square-and-multiply implementation
may keep $O(\log q)$ such registers for successive powers
and partial products; thus it uses polynomially many
qubits, for example $O(d(\log q)^2)$ with a binary
field representation. Reversible recomputation can
reduce space, but is not needed for polynomial time.
The coefficient-zero checks cost $O(d\log q)$
Boolean gates apart from reversible OR and control
logic. These gate counts are distinct from the
$\F_q$-operation bound.

\subsection{Testing the relevant constant subspace}\label{app:constant}
In a polynomial quotient, store
$w=w_0+w_1x+\cdots+w_{d-1}x^{d-1}$ with
$d=\deg g$. The statement $w\in\F_q$ is equivalent
to $w_1=\cdots=w_{d-1}=0$; it is a statement about
the unique degree-$<d$ representative and can be
tested without knowing the factors of $g$.
For each $j\ge1$, compare the binary encoding of
$w_j$ with zero into a fresh Boolean flag.
Combine these flags by a reversible OR network
and write the result to the predicate bit. Reverse
the OR network and the comparisons after the
controlled phase is applied; a constant number of
ancilla bits suffices for a linear-size OR network.

For an idempotent block $eA$ in a structure-constant
algebra, the statement to test is instead
$w\in\F_qe$. Choose an ambient coordinate
$j$ with $e_j\ne0$. If $w=ce$, comparison at
that coordinate forces $c=w_j/e_j$. Consequently
one can compute this candidate scalar and check
all $n$ equalities $w_i=ce_i$. The inverse
$e_j^{-1}\in\F_q$ is obtained classically when
the known idempotent and pivot are fixed; the
reversible circuit may use it as a field constant.
This test takes $O(n)$ field-level comparisons
and scalar multiplications and can likewise be
uncomputed after the phase. The reference element must be $e$, not
$1_A$: for a nonzero proper idempotent, $eA\cap\F_q1_A=\{0\}$, since
multiplying $c1_A\in eA$ by $1_A-e$ forces $c=0$. For example, in
$A=\F_q^3$ with $e=(1,1,0)$, the element $w=e$ is constant in $eA$ and
cannot split the block, although it does not lie in $\F_q1_A$.

Neither implementation uses a minimal-polynomial
factorization or asks which CRT coordinate has
which value. The test is about membership in a
one-dimensional, explicitly known subspace of
the current input algebra. This is the operational
content of the diagonal characterization used in
Lemmas~\ref{lem:count} and
\ref{lem:idempotents}.

\subsection{Adaptive exactness}\label{app:adaptive}
At a node, $g$ and its fixed-point basis are classical,
so the algorithm computes $s$ and selects the exact
coefficient $1/(2\sqrt{p_{q,s}})$ in
\eqref{eq:attenuation}. Conditional on \emph{any}
previous measurement history reaching this node,
Lemma~\ref{lem:prob} gives the same formula with
this node's value of $s$. The probability of a failed
measurement at the node is zero, rather than merely
bounded by an inverse polynomial. There are at most
$r-1$ nodes, and every possible history produced
by the algorithm consists of correct splits. No
union-bound error parameter is needed.

One may retain the hybrid presentation, in which
measurements separate the quantum splitting rounds
from classical gcds, or defer measurements and
reversibly compute the classical controls in a
larger uniform circuit. The latter changes the
gate-level resource accounting; the separate
quantum-arithmetic bound of Theorem~\ref{thm:main}
describes the former implementation. In either
presentation, exactness uses the input-dependent
rotation gates of Section~\ref{sec:model}.

\begin{samepage}
\subsection{Where the degree factors arise}\label{app:costs}
The following bounds use naive arithmetic and make
no claim of optimal exponents. In the general-algebra
column, multiplication uses the ambient
$n$-dimensional structure-constant presentation.
\begin{center}
\begin{tabular}{lll}
\hline
Operation & Polynomial quotient & Commutative algebra\\
\hline
One algebra multiplication
 & $O(n_P^2)$ & $O(n^3)$\\
One reversible test
 & $O(n_P^2\log q)$ & $O(n^3\log q)$\\
All quantum rounds
 & $O(n^3\log q)$ & $O(n^4\log q)$\\
Initial Frobenius basis, classical
 & $O(n^3\log q)$ & $O(n^4\log q)$\\
Classical bookkeeping
 & $O(n^4+n^3\log q)$ & $O(n^4\log q)$\\
\hline
\end{tabular}
\end{center}
\end{samepage}
The table uses the two concrete tests, with $|T|\le3$ and evaluation
length $O(\log q)$. Bookkeeping includes recomputing measured test
values, extraction, and basis updates; the classical total also includes
the one-off Frobenius computation. General effective-test costs are
stated after Theorem~\ref{thm:commutative}. In an unbalanced binary splitting tree,
there can be $n-1$ internal nodes whose parent bases
have dimensions $n,n-1,\ldots,2$. This explains why
multiplying a per-node cubic elimination cost by only
a constant would be unsound. Conversely, the quantum
test at each node depends quadratically on its
polynomial degree, so at most $n-1$ nodes with degrees
at most $n$ suffice for the advertised cubic total.

\subsection{Output verification}\label{app:verify}
Although no failing measurement branch exists in
the ideal circuit model, the factorization output
has a straightforward deterministic check. First
normalize every output polynomial to be monic
and multiply the list to confirm that the product
is the input $f$. Next confirm that each factor
$h$ is irreducible. One option uses exactly the
fixed-point criterion underlying the algorithm:
compute the kernel of $a\mapsto a^q-a$ on
$\F_q[x]/(h)$ and require its dimension to be
one. Since the input is squarefree, a reducible
output factor would have at least two distinct
irreducible factors, and its fixed-point algebra
would have dimension at least two. This check
uses only deterministic finite-field arithmetic
and linear algebra and does not rely on the
specific measured values of $\lambda$.

During the recursive run, the same criterion is
available at lower cost: the basis pushed down
to a leaf already has dimension one, and the
induction establishing its equality with the
full fixed-point algebra is
Lemma~\ref{lem:pushdown}. An independent output
checker may recompute the kernel if it does
not wish to trust the internal bookkeeping.
The distinction between internal correctness
and independent verification is useful when
the factoring routine is called by another
algebra algorithm. Exactness guarantees
correct output in the model; verification
provides a separate deterministic way to
inspect that output.

\bibliographystyle{plain}
\bibliography{polyfact}
\end{document}